\documentclass[12pt]{article}

\usepackage{fullpage,amsmath,amssymb,amsthm,float}
\usepackage{enumerate}

\def\fmod#1 #2{#1\ ({\rm mod}\ #2)}

\def\alt{{\rm alt}}

\makeatletter
\newcommand{\rmnum}[1]{\romannumeral #1}
\newcommand{\Rmnum}[1]{\expandafter\@slowromancap\romannumeral #1@}
\makeatother

\usepackage[dvips]{epsfig}
\usepackage{epsf}

\title{The non-abelian squares are not context-free}

\author{Shuo Tan\\
School of Computer Science \\
University of Waterloo \\
Waterloo, ON  N2L 3G1 \\
Canada \\
{\tt sh22tan@gmail.com} }

\theoremstyle{plain}
\newtheorem{theorem}{Theorem}

\newtheorem{lemma}[theorem]{Lemma}

\theoremstyle{definition}
\newtheorem{definition}[theorem]{Definition}
\newtheorem{example}[theorem]{Example}

\theoremstyle{remark}

\begin{document}

\maketitle


Erd\H{o}s introduced the notion of \textit{abelian square} in a 1961
paper \cite{Erdos:1961}:

\begin{definition}
A word $w$ is an \textit{abelian square} if there exist two words $s$ and $t$
such that $w = st$ and $t$ is a permutation of $s$.
\end{definition}

At the DLT 2011 conference in Milan, Italy,
Maxime Crochemore asked if the language of non-abelian squares is context-free. In this note, we answer his question.

Let $\Sigma = \{0,1\}$ be the alphabet in this context. Let $L$ denote the language of all words that are not abelian squares over $\Sigma$. Our goal is to prove the following:

\begin{theorem}
   L is not context-free.
\end{theorem}

In order to prove this, we introduce some notation. Define $w_i = 10^{i-1}$ for all $i > 1$. Let $R = w_4^{*}w_3w_2^{*}w_3w_3^{*}$.

\begin{lemma}
The word $w_4^{n}w_3w_2^{n! + n}w_3w_3^{2(n! + n)} \in L \cap {(\Sigma^2)}^{*} \cap R$.
\label{three}
\end{lemma}

\begin{proof}
Suppose $z = w_4^{n}w_3w_2^{n! + n}w_3w_3^{2(n! + n)}$. Clearly $z \in {(\Sigma^2)}^{*} \cap R$. Suppose $m$ is the total number of $1$'s in $z$. Then $m = 3n! + 4n + 2$. However, the number of $1$'s in the second half of $z$ is $\frac{4n!}{3} + 2n + 1$, which is not equal to $\frac{m}{2}$. Hence $z$ is not an abelian square. Thus $z \in L$.
\end{proof}

\begin{definition}
Let $w = 0^{s_0}10^{s_1}\cdots10^{s_k}$ be a word over the alphabet $\Sigma$. We define $$\alt(w) = | \{1 \leq i < k : s_i \neq s_{i + 1}\} |.$$
Define $\alt(\cdot)$ over a language $K$ as follows: $ \alt(K) = \max_{w \in K}{\alt(w)}$.
\end{definition}

\begin{definition}
 A sequence of non-negative integers $(a_k)_{k = 1}^n$ is called \textit{uneven} if $n > 1$ and $\exists i \in [1,n]$ such that $a_i \neq a_{i + 1}$. Here $a_{n + 1} = a_1$. Otherwise, it is called \textit{even}.
\end{definition}

\begin{definition}
Suppose $w$ is a word that contains a $1$. Then $w$ is called \textit{uneven} if the sequence $(s_1 + s_{k + 1}, s_2, \ldots, s_k)$ is uneven, where $w = 0^{s_1}10^{s_2}1\cdots0^{s_k}10^{s_{k+1}}$. Otherwise, it is called \textit{even}.
\end{definition}

\begin{example}
   Suppose $w$ is an even word and of the form $0^{s_1}10^{s_2}1\cdots0^{s_k}10^{s_{k+1}}$. Now we consider what $w$ looks like. Since $w$ is even, we get that $s_1 + s_{k + 1} = s_2 =\cdots = s_k$. Then $w = 0^s1(0^{s + t}1)^k0^t$ for some $s,t,k \geq 0$.
\label{seven}
\end{example}

\begin{lemma}
If $w$ is an uneven word, then $\alt(w^k) \geq k - 1$.
\label{eight}
\end{lemma}

\begin{proof}
Suppose $w = 0^{s_1}10^{s_2}1\cdots0^{s_k}10^{s_{k+1}}$. Then $w^k = 0^{s_1}(10^{s_2}\cdots10^{s_k}10^{s_{k + 1} + s_1})^{k - 1}10^{s_2}\cdots10^{s_k}\\10^{s_{k + 1}}$.  Since $w$ is uneven, we get $\alt(10^{s_2}\cdots10^{s_k}10^{s_{k + 1} + s_1}) \geq 1$. It follows that $\alt(w^k) \geq k - 1$.
\end{proof}

\begin{lemma}
Let $T = L \cap {(\Sigma^2)}^{*} \cap R$. Then $T$ is not context-free.
\label{nine}
\end{lemma}

\begin{proof}
(By \textit{Ogden's Lemma}) For any $n > 4$, let $z = w_4^{n}w_3w_2^{n!
+ n}w_3w_3^{2(n! + n)}$. By Lemma~\ref{three} we see that $z \in T$. Mark the
first $4n$ bits  of $z$, that is, the bits corresponding to $w_4^{n}$.
Let $m(s)$ denote the number of bits marked in $s$. Now we show by
contradiction that no decomposition $z = u_0v_0w_0x_0y_0$ satisfies all
the following three conditions:

   \begin{enumerate}
   \item{condition A: $m(v_0x_0) > 0$}
   \item{condition B: $m(v_0w_0x_0) \leq n$}
   \item{condition C: $\forall i \geq 0, u_0v_0^{i}w_0x_0^{i}y_0 \in T$}
   \end{enumerate}

Before further consideration into all possible decompositions, we first mark $z$ with different colors. Mark the bits corresponding to $w_4^n$ \textrm{red}. Mark the next 3 bits corresponding to $w_3$ \textrm{blue}. Mark the bits corresponding to $w_2^{n! + n}$ \textrm{green}. Mark the bits corresponding to $w_3w_3^{2(n! + n)}$ \textrm{black}. Define a new function $m(\textrm{color},x)$ as the number of bits in $x$ colored $\textrm{color}$. Note that $m(x)$ in our former definition is the same as $m(\textrm{red},x)$. Here is a picture of how $z$ is colored:
$$\underbrace {\text{ }{w_4}{w_4}\cdots{w_4}\text{ }}_{\textrm{red}}\underbrace{{w_3}}_{\textrm{blue}}\underbrace{\text{ }{w_2}\cdots{w_2}\text{ } }_{\textrm{green}}\underbrace{\text{ }{w_3}{w_3}\cdots{w_3}\text{ }}_{\textrm{black}}$$

Now we list all possible cases.
\begin{enumerate}[(i)]
   \item{
       Either $v$ or $x$ is the empty word. Without loss of generality, suppose $x$ is empty.
       \begin{enumerate}[(i)]
           \item{
               $v \in 0^{+}$.
           }
           \item{
               $v$ contains a $1$ and $v$ is uneven.
           }
           \item{
               $v$ contains a $1$ and $v$ is even.
           }
       \end{enumerate}
   }
   \item{
        Both $v$ and $x$ are non-empty words.
       \begin{enumerate}[(i)]
           \item{
               $v \in 0^{+}$ or $x \in 0^{+}$.
           }
           \item{
               Both $v$ and $x$ contain a $1$; $v$ is uneven or $x$ is uneven.
           }
           \item{
               Both $v$ and $x$ contain a $1$ and are even.
               \begin{enumerate}[(i)]
                   \item{
                       $m(\textrm{red},v) = 0$, which means that $w_4^n$ precedes $v$ in $z$.
                   }
                   \item{
                       $m(\textrm{red},v) > 0$.
                       \begin{enumerate}[(i)]
                           \item{
                               $m(\textrm{red},x) > 0$.
                           }
                           \item{
                               $m(\textrm{red},x) = 0$ and $m(\textrm{green},x) > 0$
                               \begin{itemize}
                                   \item{
                                       $m(\textrm{blue},x) = 3$.
                                   }
                                   \item{
                                       $m(\textrm{blue},x) = 2$.
                                   }
                                   \item{
                                       $m(\textrm{blue},x) = 1$.
                                   }
                                   \item{
                                       $m(\textrm{blue},x) = 0$ and $m(\textrm{green},x) > 1$.
                                   }
                                   \item{
                                       $m(\textrm{blue},x) = 0$ and $m(\textrm{green},x) = 1$.
                                   }
                               \end{itemize}
                           }
                           \item{
                               $m(\textrm{red},x) = m(\textrm{green},x) = 0$ and $m(\textrm{black},x) > 0$
                           }
                       \end{enumerate}
                   }
               \end{enumerate}
           }
       \end{enumerate}
   }
\end{enumerate}

Suppose there exists a decomposition $z = uvwxy$ satisfying the above three conditions simultaneously.

\medskip

Case \rmnum{1}: First we consider the case when either $v$ or $x$ is empty. Without loss of generality, suppose $x$ is empty. Then $v$ cannot be empty, since $vx$ is non-empty.

\medskip

Case \rmnum{1}.\rmnum{1}: Suppose $v = 0^{k}$ for some $k \in \mathbb{N}^{+}$. Then we select $i = 4$. Since there are more than $3$ successive $0$'s in $v^4$, this is also true for $uv^4wx^4y$. However, no word in $T$ contains more than $3$ successive $0$'s. Hence we get a contradiction.

\medskip

Case \rmnum{1}.\rmnum{2}: Suppose $v$ contains a $1$ and is uneven. We pick $i = 6$. Then $\alt(v^6) \geq 5 > \alt(R) = 4$ by Lemma~\ref{eight}. So $uv^6wx^6y \not\in T$, which violates condition C.

\medskip

Case \rmnum{1}.\rmnum{3}: Now we consider when $v$ is even. In this
case $v$ can be written in the form $0^{k}1(0^{k + s}1)^{p}0^{s}$ for
some $k,s,p \in \mathbb{N}$ (as we mentioned in Example~\ref{seven}). Then it
follows that $m(\textrm{green},v) = 0$ and $k + s = 3$ by the following
argument. Suppose $m(\textrm{green},v) > 0$. Then $m(\textrm{blue},v) =
3$. That is to say, the $w_3$ between $w_4$'s and $w_2$'s lies in $v$.
Then $v$ must be of the form $r_1 01001 r_2$ for some words $r_1$ and
$r_2$.  It follows that $k + s = 2$, since $v$ is even. Now we select
$i = 2$. Then $uv^2wx^2y = w_4^{n}w_3^{l}w_2^{n! + n}w_3w_3^{2(n! +
n)}$ for some $l > 1$, which violates condition C. Now suppose $k + s
\neq 3$. Then we pick $i = 2$. It follows that $uv^2wx^2y$ is of the
form $w_4^{l}w_{k + s + 1}^{2 + 2p}w_4^{j}w_3w_2^{n! + n}w_3w_3^{2(n! +
n)} \not \in T$, which violates condition C again. Now let $i =
\frac{n!}{1 + p}$. It follows that $z_i = uv^iwx^iy = w_4^{n! +
n}w_3w_2^{n! + n}w_3w_3^{2(n! + n)} = (w_4^{n! + n}w_3w_2^{n! +
n})(w_3w_3^{2(n! + n)})$ is an abelian square, a contradiction.
\medskip

Case \rmnum{2}: Both $v$ and $x$ are non-empty. In this case, we first
show that both $v$ and $x$ contain a $1$. Then, we show $v$ and $x$ are
even. Finally we rule out all subcases under the condition that $v$ and
$x$ are even.

\medskip

Case \rmnum{2}.\rmnum{1}: Suppose $v = 0^{k}$ or $x = 0^{l}$ for some
$k,l \in \mathbb{N}^{+}$. By a similar analysis in Case
\rmnum{1}.\rmnum{1}, we get that this case violates condition C.

\medskip

Case \rmnum{2}.\rmnum{2}: Suppose $v$ is uneven. By a similar analysis in Case \rmnum{1}.\rmnum{2}, we see that this case violates condition C. The same applies to the case when $x$ is uneven.

\medskip

Case \rmnum{2}.\rmnum{3}: Now it remains to consider when $v$ and $x$ are even. Suppose $v = 0^{k}1(0^{k + s}1)^p0^{s}$ for some $k,s,p \in \mathbb{N}$, and $x = 0^{c}1(0^{c + d}1)^e0^{d}$ for some $c,d,e \in \mathbb{N}$.

\medskip

Case \rmnum{2}.\rmnum{3}.\rmnum{1}: First of all we consider the case when $m(\textrm{red},v) = 0$. Then $m(\textrm{red},x) = 0$ since $x$ precedes $v$ in $z$. It follows that $m(\textrm{red},vx) = 0$, which violates condition A.

\medskip

Case \rmnum{2}.\rmnum{3}.\rmnum{2}: Now we turn to the case when $m(\textrm{red},v) > 0$. By the same argument in Case \rmnum{1}.\rmnum{3}, we get that $m(\textrm{green},v) = 0$ and $k + s = 3$. Note that $p < n$, for otherwise the condition $m(\textrm{green},v)  = 0$ cannot be satisfied. Now we consider the following subcases:

\medskip

Case \rmnum{2}.\rmnum{3}.\rmnum{2}.\rmnum{1}: If $m(\textrm{red},x)  > 0$, then $m(\textrm{green},x) = 0$ and $c + d = 3$. After selecting $i = \frac{n!}{2 + p + e}$, we get that $z_i = uv^iwx^iy = w_4^{n! + n}w_3w_2^{n! + n}w_3w_3^{2(n! + n)} = (w_4^{n! + n}w_3w_2^{n! + n})(w_3w_3^{2(n! + n)})$ is an abelian square, which violates condition C again.

\medskip

Case \rmnum{2}.\rmnum{3}.\rmnum{2}.\rmnum{2}: If $m(\textrm{red},x) = 0$ and $m(\textrm{green},x)  > 0$, then $m(\textrm{blue},x) < 3$, for otherwise $x$ cannot be even. There are again four subcases here.
   \begin{enumerate}
   \item{
       The first subcase is $m(\textrm{blue},x) = 2$. Then $x$ is in the form $001 r_1$, where $r_1$ is any word. We see that (a) $r_1 = \epsilon$ or (b) $r_1 = 0$, for otherwise $x$ cannot be even. (a) If $r_1 = \epsilon$, then $x = 001$. We pick $i = 2$. It follows that $uv^2wx^2y = w_4^{n + 2p + 2}w_3^{2}w_2^{n! + n}w_3w_3^{2(n! + n)} \not\in T$, which voilates condition C. (b) If $r_1 = 0$, then $x = 0010$. We pick $i = 2$. It follows that $uv^2wx^2y = w_4^{n + 2p + 2}w_3w_4w_2^{n! + n}w_3w_3^{2(n! + n)} \not\in T$, which also violates condition C.
   }
   \item{
       The second case is $m(\textrm{blue},x) = 1$. Here is a picture.
       $$z = {w_4}{w_4}\cdots{w_4}10\underbrace {\text{ }0{w_2}\cdots}_x\cdots{w_2}{w_3}\cdots{w_3}$$
       Thus (a) $x = 010$ or (b) $x = (01)^l$ for some $l >0$. (a) If $x = 010$, then we pick $i = 2$. It follows that $uv^2wx^2y = w_4^{n + 2p + 2}w_3^2w_2^{n! + n}w_3w_3^{2(n! + n)} \not\in T$. (b) Otherwise $x = (01)^l$ for some $l >0$. After picking $i = \frac{n!}{1 + p}$, we get
       \begin{align*}
           z_i &= w_4^{n! + n}w_3w_2^{n! + n + \frac{ln!}{1 + p}}w_3w_3^{2(n! + n)}\\
               &= (w_4^{n! + n}w_3w_2^{n! + n})w_2^{\frac{ln!}{1 + p}} (w_3w_3^{2(n! + n)})\\
               &= (w_4^{n! + n}w_3w_2^{n! + n}w_2^{\frac{ln!}{2(1 + p)}})(w_2^{\frac{ln!}{2(1 + p)}}w_3w_3^{2(n! + n)}).\\
       \end{align*}
       Note that $\frac{ln!}{2(1 + p)}$ is an integer since $n > 4$. Therefore $z_i$ is an abelian square. Hence $z_i \not\in T$, a contradiction.
   }
   \item{
       The third case is $m(\textrm{blue},x) = 0$ and $m(\textrm{green},x) > 1$. Similarly $x$ must be of the form $(01)^l$ or $(10)^l$ for some $l \in \mathbb{N}^{+}$, since $x$ is even. We pick $i = \frac{n!}{1 + p}$ and find the same result as in the second case.
   }
   \item{
       The last case is exactly when $m(\textrm{blue},x) = 0$ and $m(\textrm{green},x) = 1$. Then $x$ has to be the first or the last letter of the substring $w_2^{n! + n}$ of $z$, since $x$ cannot be a single $0$ or $1$ (there are no successive $1$'s in elements of $R$). Moreover, we find that $x$ cannot be the first letter of $w_2^{n! + n}$, since $m(\textrm{blue},x) = 0$. It follows that $x$ is the trailing $0$ of $w_2^{n! + n}$. Under this circumstance, we find that $x = 0(100)^e10$, since $x$ is even. Now we select $i = \frac{n!}{1 + p}$. It follows that
       \begin{align*}
       z_i&= uv^iwx^iy \\
          &= w_4^{n! + n}w_3w_2^{n! + n}w_3w_3^{2(n! + n) + \frac{(1 + e)n!}{1 + p}} \\
          &= (w_4^{n! + n}w_3w_2^{n! + n})w_3^{\frac{(1 + e)n!}{1 + p}}(w_3w_3^{2(n! + n)})\\
          &= (w_4^{n! + n}w_3w_2^{n! + n}w_3^{\frac{(1 + e)n!}{2(1 + p)}})(w_3^{\frac{(1 + e)n!}{2(1 + p)}}w_3w_3^{2(n! + n)}).\\
    \end{align*}
        Note that $\frac{(1 + e)n!}{2(1 + p)}$ is an integer since $n > 4$. Thus $z_i$ is an abelian square. Therefore $z_i \not\in T$, a contradiction.
   }
   \end{enumerate}

Case \rmnum{2}.\rmnum{3}.\rmnum{2}.\rmnum{3}: The last possible subcase is when $m(\textrm{red},x) = 0$ and $m(\textrm{green},x)  = 0$ and $m(\textrm{black},x)  > 0$. In this case, we get that $c + d = 2$. Now we pick $i = \frac{n!}{1 + p}$. It follows that
   \begin{align*}
       z_i&= uv^iwx^iy \\
          &= w_4^{n! + n}w_3w_2^{n! + n}w_3w_3^{2(n! + n) + \frac{(1 + e)n!}{1 + p}} \\
          &= (w_4^{n! + n}w_3w_2^{n! + n})w_3^{\frac{(1 + e)n!}{1 + p}}(w_3w_3^{2(n! + n)})\\
          &= (w_4^{n! + n}w_3w_2^{n! + n}w_3^{\frac{(1 + e)n!}{2(1 + p)}})(w_3^{\frac{(1 + e)n!}{2(1 + p)}}w_3w_3^{2(n! + n)}).\\
    \end{align*}
Note that $\frac{(1 + e)n!}{2(1 + p)}$ is an integer since $n > 4$. Thus $z_i$ is an abelian square. Therefore $z_i \not\in T$, a contradiction again.

With the above discussion, we claim that no decomposition of $z$ can satisfy all three conditions simultaneously. Thus $T$ is not context-free.
\end{proof}

\begin{proof}
Combining Lemma~\ref{nine} with the closure property, that the language obtained by intersecting a context-free language and a regular language is still context-free, we finally obtain this conclusion.
\end{proof}

\end{document}